\documentclass[12pt,letterpaper]{article}
\usepackage{amsmath,amsfonts,amsthm,amssymb}

\theoremstyle{plain}

\numberwithin{equation}{section}
\newtheorem{thm}{Theorem}[section]
\newtheorem{lem}[thm]{Lemma}
\newtheorem{cor}[thm]{Corollary}

\newenvironment{exam}[1]%
{\begin{flushleft}\textbf{Example #1}.\enspace}%
{\end{flushleft}}

\newcounter{cond}

\newcommand{\complex}{{\mathbb C}}
\newcommand{\positive}{{\mathbb N}}
\newcommand{\real}{{\mathbb R}}
\newcommand{\ascript}{{\mathcal A}}

\newcommand{\cscript}{{\mathcal C}}
\newcommand{\pscript}{{\mathcal P}}
\newcommand{\qscript}{{\mathcal Q}}
\newcommand{\sscript}{{\mathcal S}}
\newcommand{\rmtr}{\mathrm{tr}}
\newcommand{\rmcyl}{\mathrm{cyl}}
\newcommand{\dhat}{\widehat{D}}
\newcommand{\fhat}{\widehat{f}}
\newcommand{\ghat}{\widehat{g}}
\newcommand{\hhat}{\widehat{h}}

\newcommand{\chihat}{\widehat{\chi}}
\newcommand{\mutilde}{\widetilde{\mu}}

\newcommand{\ab}[1]{\left|#1\right|}
\newcommand{\doubleab}[1]{\left\|#1\right\|}
\newcommand{\brac}[1]{\left\{#1\right\}}
\newcommand{\paren}[1]{\left(#1\right)}
\newcommand{\sqbrac}[1]{\left[#1\right]}
\newcommand{\elbows}[1]{{\left\langle#1\right\rangle}}
\newcommand{\ket}[1]{{\left|#1\right>}}
\newcommand{\bra}[1]{{\left<#1\right|}}

\begin{document}

\title{DISCRETE QUANTUM GRAVITY
}
\author{S. Gudder\\ Department of Mathematics\\
University of Denver\\ Denver, Colorado 80208, U.S.A.\\
sgudder@du.edu
}
\date{}
\maketitle

\begin{abstract}
We discuss the causal set approach to discrete quantum gravity. We begin by describing a classical sequential growth process in which the universe grows one element at a time in discrete steps. At each step the process has the form of a causal set and the ``completed'' universe is given by a path through a discretely growing chain of causal sets. We then introduce a method for quantizing this classical formalism to obtain a quantum sequential growth process which may lead to a viable model for a discrete quantum gravity. We also give a method for quantizing random variables in the classical process to obtain observables in the corresponding quantum process. The paper closes by showing that a discrete isometric process can be employed to construct a quantum sequential growth process.
\end{abstract}

\section{Introduction}  % Section 1
This paper explores the causal set approach to discrete quantum gravity \cite{blms87, bdghs03, hen09}. There are many good review articles on this subject \cite{hen06, sor03, sur11} and we refer the reader to these works for more details and motivation. The origins of this approach stem from studies of the causal structure $(M,<)$ of a Lorentzian spacetime $(M,g)$. For $a,b\in M$ we write $a<b$ if $b$ is in the causal future of $a$. If there are no closed causal curves in $(M,g)$, then $(M,<)$ is a partially ordered set (poset). That is, the order $<$ satisfies:
\begin{list} {(\arabic{cond})}{\usecounter{cond}
\setlength{\rightmargin}{\leftmargin}}
%(1)
\item $a\not< a$ for all $a\in M$ (irreflexivity).
%(2)
\item $a<b$ and $b<c$ implies that $a<c$ (transitivity).
\end{list}
It has been shown that $(M,<)$ possesses much of the information contained in $(M,g)$
\cite{bms89, mal77, mp04, zee64}. In fact, $(M,<)$ determines the topological and even the differential structure of the manifold $(M,g)$. Moreover, $(M,<)$ can be employed to find the length of line elements and the dimension of
$(M,g)$. Finally, counting arguments on $(M,<)$ can be employed to find volume elements in $(M,g)$. Because of these properties, it is believed that the order structure $(M,<)$ provides a viable candidate for describing a discrete quantum gravity.

For a poset $(A,<)$, the \textit{past} of $b\in A$ is $\brac{a\in A\colon a<b}$. We say that $\paren{A,<}$ is
\textit{past finite} if the past of $b$ has finite cardinality for every $b\in A$. A \textit{causal set} is a past finite countable poset. One of the simplifications of this paper is that the relevant posets considered will be finite so they are automatically causal sets. Another simplification is that we shall only consider unlabeled posets. In the literature, causal sets are usually labeled according to the order of ``birth'' and this causes complications because covariant properties are independent of labeling \cite{blms87, bdghs03, rs00, sor03}. In this way our causal sets are automatically covariant.

Section~2 describes a classical sequential growth process in which the universe grows one element at a time in discrete steps. At each step the process has the form of a causal set and the ``completed'' universe is given by a path through a discretely growing chain of causal sets. The transition probability at each step is given by an expression due to Rideout-Sorkin \cite{rs00, vr06}. Letting $\Omega$ be the set of paths, $\ascript$ be the $\sigma$-algebra generated by cylinder sets and $\nu$ the probability measure governed by the transition probabilities, the dynamics is described by a Markov chain on the probability space $(\Omega ,\ascript ,\nu )$.

In Section~3 we quantize classical frameworks by forming the Hilbert space $H=L_2(\Omega ,\ascript ,\nu )$. The quantum dynamics is given by a sequence of states $\rho _n$ on $H$ that satisfy a consistency condition. We employ $\rho _n$ to construct decoherence functionals and a quantum measure $\mu$ on a ``quadratic algebra''
$\sscript$ of subsets of $\Omega$. In general, the set $\sscript$ is strictly between the collection of cylinder sets and
$\ascript$. We then present $(\Omega ,\sscript ,\mu )$ as a candidate model for quantum gravity. We also give a method for quantizing random variables in the classical process to obtain observables in the corresponding quantum process. This quantization is then used to define a quantum integral.

The sequence of states $\rho _n$ discussed in Section~3 is called a quantum sequential growth process. Section~4 shows that a discrete isometric process can be employed to construct a quantum sequential growth process. This work is related to the ``sum over histories'' approach to quantum mechanics \cite{hen06}.

Of course, much work remains to be done. Of primary importance is to find the specific form of the classical measure $\nu$ and the quantum measure $\mu$. One guiding principle is that classical general relativity theory should be a
``good approximation'' to this quantum counterpart.

\section{Quantum Sequential Growth Processes} % Section 2
Let $\pscript _n$ be the collection of all posets of cardinality $n$, $n=0,1,2,\ldots$, and let $\pscript =\cup\pscript _n$
be the collection of all finite posets. An element $a\in x$ for $x\in\pscript$ is \textit{maximal} if there is no $b\in x$ with $a<b$. If $x\in\pscript _n$, $y\in \pscript _{n+1}$, then $x$ \textit{produces} $y$ if $y$ is obtained from $x$ by adjoining a single maximal element to $x$. We also say that $x$ is a \textit{producer} of $y$ and $y$ is an
\textit{offspring} of $x$. If $x$ produces $y$ we write $x\to y$. Of course, $x$ may produce many offspring and a poset may be the offspring of many producers. Also, $x$ may produce $y$ in various isomorphic ways. For example, in Figure~1 we have that $x$ produces $u,v,w$. In this paper we identify isomorphic copies of a poset so we identify $u, v,w$ and say that the multiplicity of $x\to u$ is three and write $m(x\to u)=3$. (Strictly speaking, the multiplicity requires a labeling of the elements of a poset and this is the only place that labeling needs to be mentioned.) In Figure~1, notice that within each circle is a Hasse diagram of a poset and a rising line in a diagram represents a \textit{link}.

% Figure 1
\begin{figure}
\setlength{\unitlength}{8pt}
\begin{center}
\begin{picture}(45,12)
\put(22,8){\circle{7}}
\put(20.5,8){\circle*{.35}}
\put(22,8){\circle*{.35}}
\put(23.5,8){\circle*{.35}}
\put(22,4){\makebox{$x$}}
\put(22,10.5){\vector(0,1){4}}
\put(19.5,8){\vector(-1,1){7.5}}
\put(24.5,8){\vector(1,1){7.5}}
\put(10,17){\circle{7}}
\put(8.5,16.2){\circle*{.35}}
\put(10,16.2){\circle*{.35}}
\put(11.5,16.2){\circle*{.35}}
\put(8.5,16.2){\line(1,2){.8}}
\put(10,16.2){\line(-1,2){.8}}
\put(9.23,17.85){\circle*{.35}}
\put(6,16){\makebox{$u$}}
\put(22,17){\circle{7}}
\put(20.5,16.2){\circle*{.35}}
\put(22,16.2){\circle*{.35}}
\put(23.5,16.2){\circle*{.35}}
\put(23.5,16.2){\line(-1,2){.8}}
\put(22,16.2){\line(1,2){.8}}
\put(22.75,17.8){\circle*{.35}}
\put(18,16){\makebox{$v$}}
\put(34,17){\circle{7}}
\put(32.5,16.2){\circle*{.35}}
\put(34,16.2){\circle*{.35}}
\put(35.5,16.2){\circle*{.35}}
\put(32.5,16.2){\line(1,1){1.5}}
\put(35.5,16.2){\line(-1,1){1.5}}
\put(34,17.7){\circle*{.35}}
\put(30,16){\makebox{$w$}}
\end{picture}
%\caption{} %  figure out about caption and [htbp]
\vskip -2pc
\centerline{\textbf{Figure 1}}
%\label{fig1}
\end{center}
\end{figure}

The transitive closure of ${}\to$ makes $\pscript$ into a poset itself. A \textit{path} in $\pscript$ is a sequence (string)
$x_0x_1x_2\cdots$ where $x_i\in\pscript _i$ and $x_i\to x_{i+1}$, $i=0,1,2,\ldots\,$. An $n$-\textit{path} in $\pscript$ is a finite string $x_0x_1\cdots x_n$ where again $x_i\in\pscript _i$ and $x_i\to x_{i+1}$. We denote the set of paths by
$\Omega$ and the set of $n$-paths by $\Omega _n$. If $a,b\in x$ with $x\in\pscript$, we say that $a$ is an
\textit{ancestor} of $b$ and $b$ is a \textit{successor} of $a$ if $a<b$. We say that $a$ is a \textit{parent} of $b$ and
$b$ is a \textit{child} of $a$ if $a<b$ and there is no $c$ with $a<c<b$. A link in a Hasse diagram represents a
parent-child relationship.

Let $t=(t_0,t_1,\ldots )$ be a sequence of nonnegative numbers (called \textit{coupling constants} \cite{rs00, sor03}). For $r,s\in\positive$ with $r\le s$, define
\begin{equation*}
\lambda _t(s,r)=\sum _{k=r}^s\binom{s-r}{k-r}t_k =\sum _{k=0}^{s-r}\binom{s-r}{k}t_{r+k}
\end{equation*}
For $x\in\pscript _n$, $y\in\pscript _{n+1}$ with $x\to y$ we define the \textit{transition probability}
\begin{equation}         % equation (2.1)
\label{eq21}
p_t(x\to y)=m(x\to y)\frac{\lambda _t(\alpha ,\pi )}{\lambda _t (n,0)}
\end{equation}
where $\alpha$ is the number of ancestors and $\pi$ the number of parents of the adjoined maximal element to
$x$ that produces $y$. The definition of $p_t(x\to y)$ originally appears in \cite{rs00}. It is shown there that
$p_t(x\to y)$ is a probability distribution in that it satisfies the Markov-sum rule
\begin{equation*}
\sum\brac{p_t(x\to y)\colon y\in\pscript _{n+1}\hbox{ with }x\to y}=1
\end{equation*}
The distribution $y\mapsto\pscript _t(x,y)$ is essentially the most general that is consistent with principles of causality and covariance \cite{rs00, sor03}. It is hoped that other theoretical principles or experimental data will determine the coupling constants. One suggestion is to take $t_k=1/k!$ \cite{sor03}.

As an illustration, which probably will not work for quantum gravity and cosmology, let $t_k=t^k$ for some $t>0$. This case has been previously studied and is called a \textit{percolation dynamics} \cite{hen06, sor03}. For this choice we have
\begin{equation*}
\lambda _t(s,r)=\sum _{k=0}^{s-r}\binom{s-r}{k}t^{r+k}=t^r\sum _{k=0}^{s-r}\binom{s-r}{k}t^k=t^r(1+t)^{s-r}
\end{equation*}
and as a special case $\lambda _t (n,0)=(1+t)^n$. Letting $\beta$ be the number of elements of $x$ not related to the adjoined maximal element, by \eqref{eq21} we have
\begin{equation*}
p_t(x\to y)=m(x\to y)t^\pi\frac{(1+t)^{\alpha -\pi}}{(1+t)^n}=m(x\to y)\frac{t^\pi}{(1+t)^{\pi +\beta}}
\end{equation*}
Letting $r=t(1+t)^{-1}$ we have that $1-r=(1+t)^{-1}$ and we obtain the more familiar form
$p_t(x\to y)=m(x\to y)r^\pi(1-r)^\beta$

We call an element $x\in\pscript$ a \textit{site} and view a site $x\in\pscript _n$ as a possible universe at step $n$ while a path may be viewed as a possible (evolved) universe. The set $\pscript$ together with the set of transition probabilities $p_t (x\to y)$ forms a \textit{classical sequential growth process} (CSGP) which we denote by
 $(\pscript ,p_t)$ \cite{rs00, vr06}. It is clear that $(\pscript ,p_t)$ is a Markov chain. (In traditional Markov chains, sites are called states but we reserve that term for quantum states to be used later.) As with any Markov chain, the probability of an $n$-path $\omega =x_0x_1\cdots x_n$ is 
\begin{equation*}
p_t^n(\omega )=p_t(x_0\to x_1)p_t(x_1\to x_2)\cdots p_t(x_{n-1}\to x_n)
\end{equation*}
and the probability of a site $x\in\pscript _n$ is
\begin{equation*}
p_t^n(x)=\sum\brac{p_t^n(\omega )\colon\omega\in\Omega _n,x_n=x}
\end{equation*}
Of course, $\omega\mapsto p _t^n(\omega )$ is a probability measure on $\Omega _n$.

\begin{exam}{1}
Figure~2 illustrates the first three steps of a CSGP where the 2 indicates the multiplicity $m(x_3\to x_6)$. To compute probabilities, we need the values of $\lambda _t(\alpha ,\pi )$ given in Table~1.
\end{exam}

% Figure 2
\setlength{\unitlength}{8pt}
\begin{center}
\begin{picture}(45,26)
\put(22,-1){\circle{7}}  % bottom circle - x_0
\put(21.5,-1.5){\makebox{$\emptyset$}}
\put(17.5,-1.5){\makebox{$x_0$}}
\put(22,1.5){\vector(0,1){4}} % vector between x_0 and x_1
\put(22,8){\circle{7}}  % second row circle - x_1
\put(22,8){\circle*{.35}}
\put(17.5,7){\makebox{$x_1$}}
\put(19.5,8){\vector(-1,1){4}} % vector between x_1 and x_2
\put(24.5,8){\vector(1,1){4}} % vector between x_1 and x_3
\put(14,14){\circle{7}} % third row left circle - x_2
\put(12,15.5){\vector(-1,1){5.5}} % vector between x_2 and x_4
\put(16,15,5){\vector(1,1){5}} % vector between x_2 and x_6
\put(14,15){\circle*{.35}}
\put(14,13){\circle*{.35}}
\put(14,13.2){\line(0,1){1.5}}
\put(9.75,13.75){\makebox{$x_2$}}
\put(14,16.5){\vector(0,1){4}} % vector between x_2 and x_5
\put(30,14){\circle{7}} % third row right circle - x_3
\put(28,15.5){\vector(-1,1){5}} % vector between x_3 and x_6
\put(32,15,5){\vector(1,1){5}} % vector between x_3 and x_8
\put(24,17){\makebox{$2$}}
\put(29,14){\circle*{.35}}
\put(31,14){\circle*{.35}}
\put(33,13.5){\makebox{$x_3$}}
\put(30,16.5){\vector(0,1){4}} % vector between x_3 and x_7
\put(5,23){\circle{7}} % fourth row circle - x_4
\put(5,21.5){\circle*{.35}}
\put(5,23){\circle*{.35}}
\put(5,24.5){\circle*{.35}}
\put(5,21.5){\line(0,1){3}}
\put(.75,23){\makebox{$x_4$}}
\put(14,23){\circle{7}} % fourth row circle - x_5
\put(14,22){\circle*{.35}}
\put(13.25,23.5){\circle*{.35}}
\put(14.75,23.5){\circle*{.35}}
\put(14,22){\line(1,2){.8}}
\put(14,22){\line(-1,2){.8}}
\put(10,23){\makebox{$x_5$}}
\put(22,23){\circle{7}} % fourth row circle - x_6
\put(21,22){\circle*{.35}}
\put(22.5,22){\circle*{.35}}
\put(21,23.5){\circle*{.35}}
\put(21,22){\line(0,1){1.5}}
\put(18,23){\makebox{$x_6$}}
\put(30,23){\circle{7}} % fourth row circle - x_7
\put(29,22){\circle*{.35}}
\put(31,22){\circle*{.35}}
\put(30,23.75){\circle*{.35}}
\put(29,22){\line(1,2){.8}}
\put(31,22){\line(-1,2){.8}}
\put(26,23){\makebox{$x_7$}}
\put(38,23){\circle{7}} % fourth row circle - x_8
\put(36.5,23){\circle*{.35}}
\put(38,23){\circle*{.35}}
\put(39.5,23){\circle*{.35}}
\put(41,23){\makebox{$x_8$}}
\end{picture}
%\caption{} %  figure out about caption and [htbp]
%\label{fig2}
\vskip 3pc
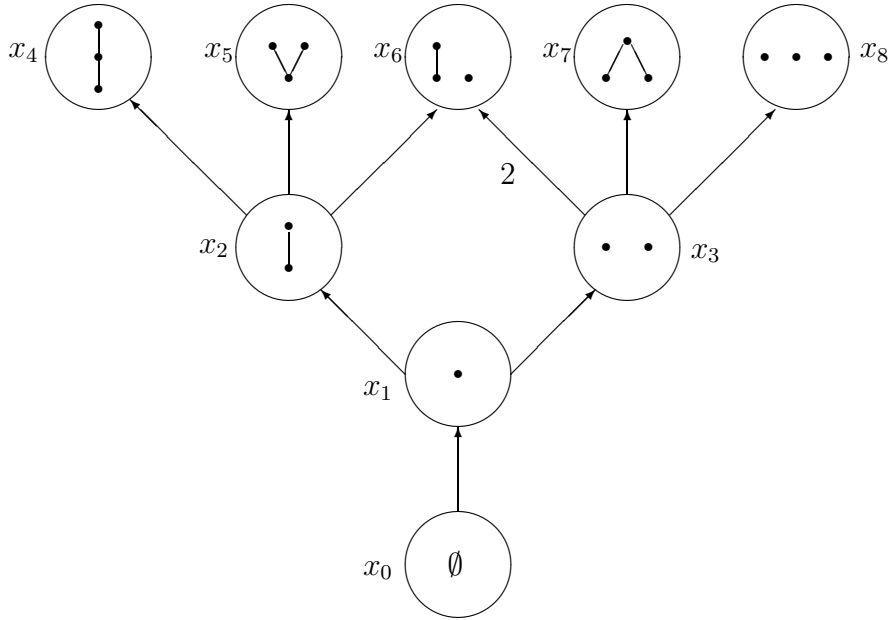
\centerline{\textbf{Figure 2}}
\end{center}

\begin{center}
\begin{tabular}{c|c|c|c|c|c|c}
$(\alpha ,\pi)$&$(0,0)$&$(1,0)$&$(1,1)$&$(2,0)$&$(2,1)$&$(2,2)$\\
\hline
$\lambda _t(\alpha ,\pi)$&$t_0$&$t_0+t_1$&$t_1$&$t_0+2t_1+t_2$&$t_1+t_2$&$t_2$\\
\noalign{\medskip}
\multicolumn{7}{c}%
{\textbf{Table 1}}\\
\end{tabular}
\end{center}
\vglue 3pc

\noindent From Table~1 and \eqref{eq21} we obtain the transition probabilities given in Table~2 where $s_0=t_0+t_1$,
$s_1=t_0+2t_1+t_2$
\vglue 3pc

%\bigskip\parindent=0pt
\begin{center}
\begin{tabular}{c|c|c|c|c|c}
$x_i\to x_j$&$x_0\to x_1$&$x_1\to x_2$&$x_1\to x_3$&$x_2\to x_4$&$x_2\to x_5$\\
\hline
$p_t(x_i\to x_j)$&$1$&$t_1/s_0$&$t_0/s_0$&$(t_1+t_2)/s_1$&$t_1/s_1$\\
\end{tabular}
\end{center}
\vskip 2pc
\begin{center}
\begin{tabular}{c|c|c|c|c}
$x_i\to x_j$&$x_2\to x_6$&$x_3\to x_6$&$x_3\to x_7$&$x_3\to x_8$\\
\hline
$p_t(x_i\to x_j)$&$t_0/s_1$&$2t_1/s_1$&$t_2/s_1$&$t_0/s_1$\\
\noalign{\bigskip}
\multicolumn{5}{c}%
{\textbf{Table 2}}\\
\end{tabular}
\end{center}
%\bigskip\parindent=18pt
\vglue 3pc

\noindent Finally, Table~3 lists the probabilities of the various sites, where $s_2=s_0s_1$ and $p_t^0(x_0)=1$ by convention.
\vglue 3pc

%\bigskip\parindent=0pt   % figure out how to put more space below hline
\begin{center}
\begin{tabular}{c|c|c|c|c|c|c|c|c|c}
$x_i$&$x_0$&$x_1$&$x_2$&$x_3$&$x_4$&$x_5$&$x_6$&$x_7$&$x_8$\\
\hline
$p_t^{(n)}(x_i)$&$1$&$1$&$\frac{t_1}{s_0}$&$\frac{t_0}{s_0}$&$\frac{t_1(t_1+t_2)}{s_2}$%
&$\frac{t_1^2}{s_2}$&$\frac{3t_0t_1}{s_2}$&$\frac{t_0t_2}{s_2}$&$\frac{t_0^2}{s_2}$\\
\noalign{\bigskip}
\multicolumn{10}{c}%
{\textbf{Table 3}}\\
\end{tabular}
\end{center}
%\bigskip\parindent=18pt
\vskip 3pc

\begin{exam}{2}
Figure~3 illustrates the offspring of $x_8$ in Figure~2.
\end{exam}

% Figure 3
\setlength{\unitlength}{8pt}
\begin{center}
\begin{picture}(45,10)
\put(22,8){\circle{7}}  % bottom circle x_8
\put(20.5,8){\circle*{.35}}
\put(22,8){\circle*{.35}}
\put(23.5,8){\circle*{.35}}
\put(22,4){\makebox{$x_8$}}
\put(19.5,8){\vector(-1,1){7.5}} % vector between x_8 and x
\put(13.5,12){\makebox{$3$}}
\put(24.5,8){\vector(1,1){7.5}} % vector between x_8 and u
\put(20.5,10){\vector(-1,2){2.25}} % vector between x_8 and y
\put(18,12){\makebox{$3$}}
\put(23.5,10){\vector(1,2){2.25}} % vector between x_8 and z
\put(10,17){\circle{7}} % top row circle x
\put(8.5,16.2){\circle*{.35}}
\put(10,16.2){\circle*{.35}}
\put(11.5,16.2){\circle*{.35}}
\put(8.5,16.2){\line(1,2){.8}}
\put(10,16.2){\line(-1,2){.8}}
\put(9.23,17.85){\circle*{.35}}
\put(6,16){\makebox{$x$}}
\put(18,17){\circle{7}}  % top row circle y
\put(16.5,16.2){\circle*{.35}}
\put(18,16.2){\circle*{.35}}
\put(19.5,16.2){\circle*{.35}}
\put(16.5,17.8){\circle*{.35}}
\put(16.5,16.2){\line(0,1){1.5}}
\put(14.5,16){\makebox{$y$}}
\put(26,17){\circle{7}}  % top row circle z
\put(24.5,16.2){\circle*{.35}}
\put(26,16.2){\circle*{.35}}
\put(27.5,16.2){\circle*{.35}}
\put(26,17.7){\circle*{.35}}
\put(24.5,16.2){\line(1,1){1.5}}
\put(27.5,16.2){\line(-1,1){1.5}}
\put(26,16.2){\line(0,1){1.5}}
\put(22.5,16){\makebox{$z$}}
\put(34,17){\circle{7}} % top row circle u
\put(32.25,16.75){\circle*{.35}}
\put(33.45,16.75){\circle*{.35}}
\put(34.55,16.75){\circle*{.35}}
\put(35.75,16.75){\circle*{.35}}
\put(30,16){\makebox{$u$}}
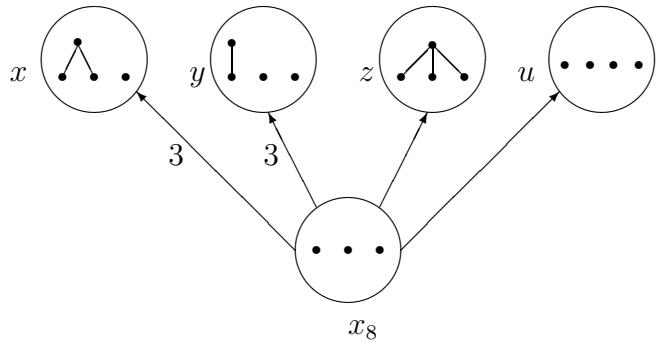
\centerline{\textbf{Figure 3}}
\end{picture}
\end{center}
\vskip 3pc

We now compute the transition probabilities.
\begin{align*}
p_t(x_8\to x)&=\frac{3\lambda _t(2,2)}{\lambda _t(3,0)}=\frac{3t_2}{t_0+3t_1+3t_2+t_3}\\
\noalign{\bigskip}
p_t(x_8\to y)&=\frac{3\lambda _t(1,1)}{\lambda _t(3,0)}=\frac{3t_1}{t_0+3t_1+3t_2+t_3}\\
\noalign{\bigskip}
p_t(x_8\to z)&=\frac{3\lambda _t(3,3)}{\lambda _t(3,0)}=\frac{t_3}{t_0+3t_1+3t_2+t_3}\\
\noalign{\bigskip}
p_t(x_8\to u&=\frac{3\lambda _t(0,0)}{\lambda _t(3,0)}=\frac{t_0}{t_0+3t_1+3t_2+t_3}\\
\end{align*}

Letting $\ascript _n$ be the power set $2^{\Omega _n}$ we have that $\ascript _n$ is an algebra of subsets of
$\Omega _n$ and $(\Omega _n,\ascript _n,p _t^n)$ is a probability space. Now we can consider $\Omega _n$ to be the product space $\Omega _n=\pscript _0\times\pscript _1\times\cdots\times\pscript _n$ and $\Omega$ to be the product space $\Omega =\pscript _0\times\pscript _1\times\pscript _2\times\cdots\,$. (Strictly speaking $\Omega _n$ is a subset of $\pscript _0\times\pscript _1\times\cdots\times\pscript _n$ because all elements of the latter set do not correspond to $n$-paths. However, we can define
\begin{equation*}
p _t^n(\pscript _0\times\pscript _1\times\cdots\times\pscript _n\smallsetminus\Omega _n)=0
\end{equation*}
and adjoining sets of measure zero is harmless. The same remark holds for $\Omega$ and for cylinder sets to be discussed next.) A subset $C\subseteq\Omega$ is a \textit{cylinder set} if
\begin{equation}         % equation (2.2)
\label{eq22}
C=C_1\times\pscript _{n+1}\times\pscript _{n+2}\times\cdots
\end{equation}
for some $C_1\in\ascript _n$. In particular, if $\omega\in\Omega _n$, then the \textit{elementary cylinder set}
$\rmcyl (\omega )$ is defined by
\begin{equation*}
\rmcyl (\omega )=\omega\times\pscript _{n+1}\times\pscript _{n+2}\times\cdots
\end{equation*}
It is easy to check that the collection of cylinder sets $\cscript (\Omega )$ forms an algebra of subsets of $\Omega$. Moreover, for $C\in\cscript (\Omega )$ of the form \eqref{eq22} we define $p_t(C)=p_t^n(C_1)$. Then $p_t$ is a well-defined probability measure on the algebra $\cscript (\Omega )$. It follows from the Kolmogorov extension theorem that $p_t$ has a unique extension to a probability measure $\nu _t$ on the $\sigma$-algebra $\ascript$ generated by
$\cscript (\Omega )$. We conclude that $(\Omega ,\ascript ,\nu _t)$ is a probability space. We can identify
$\ascript _n$ with the algebra of cylinder sets of the form \eqref{eq22} to obtain an increasing sequence of subalgebras $\ascript _0\subseteq\ascript _1\subseteq\cdots$ of $\ascript$ that generate $\ascript$. Also, the restriction $\nu _t\mid\ascript _n=p_t^n$. 

\section{Quantum Sequential Growth Processes} % Section 3
In Section~2 we described a general CSPG $(\pscript ,p _t)$. We now show how to ``quantize'' $(\pscript ,p_t)$ to obtain a quantum sequential growth process (QSGP). It is hoped that this formalism can be employed to construct a model for discrete quantum gravity. At the end of Section~2 we formed a path probability space
$(\Omega ,\ascript ,\nu _t)$ which we interpret as a space of potential universes. Let
$H=L_2(\Omega ,\ascript ,\nu _t)$ be the \textit{path Hilbert space}. We previously observed that $\ascript _n$ considered as an algebra of cylinder sets is a subalgebra of $\ascript$ and $\nu _t\mid\ascript _n=p_t^n$. We conclude that the $n$-\textit{path Hilbert spaces} $H_n=L_2(\Omega ,\ascript _n,p_t^n)$ form an increasing sequence $H_1\subseteq H_2\subseteq\cdots$ of closed subspaces of $H$. Assuming that $p_t^n(\omega )\ne 0$ for every
$\omega\in\Omega _n$ we have that $\dim (H_n)=\ab{\Omega _n}$ and that
\begin{equation*}
\brac{\chi _{\rmcyl (\omega )}/p_t^n(\omega )^{1/2}\colon\omega\in\Omega _n}
\end{equation*}
forms an orthonormal basis for $H_n$.

Let $\rho$ be a state (density operator) on $H_n$. We can and shall assume that $\rho$ is also a state on $H$ by defining $\rho f=0$ for all $f\in H_n^\perp$. If $A\in\ascript _n$ then the characteristic function $\chi _A\in H$ and
$\doubleab{\chi _A}=p_t^n(A)^{1/2}$. We define the \textit{decoherence functional}
$D_\rho\colon\ascript _n\times\ascript _n\to\complex$ by
\begin{equation*}
D_\rho (A,B)=\rmtr\paren{\rho\ket{\chi _B}\bra{\chi _A}}
\end{equation*}
It can be shown \cite{gud111} that $D_\rho$ has the usual properties of a decoherence functional. Namely,
$D_\rho (A,B)=\overline{D_\rho (B,A)}$, $A\mapsto D_\rho (A,B)$ is a complex measure on $\ascript _n$ and if
$A_i\in\ascript _n$, $i=1,\ldots ,r$, then the $r\times r$ matrix with components $D_\rho (A_i,A_j)$ is positive semidefinite. We also define the $q$-\textit{measure} $\mu _\rho\colon\ascript _n\to\real ^+$ by
$\mu _\rho (A)=D_\rho (A,A)$. In general, $\mu _\rho$ is not additive so $\mu _\rho$ is not a measure on
$\ascript _n$. However, $\mu _\rho$ is \textit{grade}-2 \textit{additive} \cite{gt09, gud09, sor94, sor07} in the sense that if $A,B,C\in\ascript _n$ are mutually disjoint then

\begin{align}         % equation (3.1)
\label{eq31}
\mu _\rho (A\cup B\cup C)&=\mu _\rho (A\cup B)+\mu _\rho (A\cup C)+\mu _\rho (B\cup C)\notag\\
  &\quad -\mu _\rho (A)-\mu _\rho (B)-\mu _\rho (C)
\end{align}

A subset $\qscript\subseteq\ascript$ is a \textit{quadratic algebra} if $\emptyset ,\Omega\in\qscript$ and if
$A,B,C\in\qscript$ are mutually disjoint with $A\cup B,A\cup C,B\cup C\in\qscript$, then $A\cup B\cup C\in\qscript$. A $q$-\textit{measure} on a quadratic algebra $\qscript$ is a map $\mu\colon\qscript\to\real ^+$ satisfying \eqref{eq31} whenever $A,B,C\in\qscript$ are mutually disjoint with $A\cup B,A\cup C,B\cup C\in\qscript$. In particular $\ascript _n$ is a quadratic algebra and $\mu _\rho\colon\ascript _n\to\real ^+$ is a $q$-measure in this sense.

Let $\rho _n$ be a state on $H_n$, $n=1,2,\ldots$, which can be viewed as a state on $H$. We say that the sequence $\rho _n$ is \textit{consistent} if
\begin{equation*}
D_{\rho _{n+1}}(A\times\pscript _{n+1}, B\times\pscript _{n+1})=D_{\rho _n}(A,B)
\end{equation*}
for every $A\in\ascript _n$. We call a consistent sequence $\rho _n$ a \textit{discrete quantum process} and we call the $\rho _n$ the \textit{local states} for the process. If $\lim\rho _n=\rho$ exists in the strong operator topology, we call $\rho$ the \textit{global} state for the process. If the global state $\rho$ exists, then $\mu _\rho$ is a (continuous)
$q$-measure on $\ascript$ that extends $\mu _{\rho _n}$, $n=1,2,\ldots\,$. Unfortunately the global state does not exist, in general, so we must work with the local states \cite{gud112, gsapp, sor11}. We contend that there is a discrete quantum process $\rho _n$ on the path Hilbert space $H$ that describes the dynamics for a discrete quantum gravity. As with the probability measures $p _t^n$, theoretical principles or experimental data will be required to give restrictions on the possible $\rho _n$. We shall consider one possibility shortly.

Let $\rho _n$ be a discrete quantum process on $H=L_2(\Omega ,\ascript ,\nu _t)$. If $C\in\cscript (\Omega )$ has the form \eqref{eq22} we define $\mu (C)=\mu _{\rho _n}(C_1)$. It is easy to check that $\mu$ is well-defined and gives a $q$-measure on algebra $\cscript (\Omega )$. In general, $\mu$ cannot be extended to a $q$-measure on $\ascript$, but it is important to extend $\mu$ to other physically relevant sets \cite{gud112, mocs05, sor11}. We say that a set
$A\in\ascript$ is \textit{suitable} if $\lim\rmtr\paren{e_n\ket{\chi _A}\bra{\chi _A}}$ exists and is finite and if this is the case we define $\mutilde (A)$ to be the limit. We denote the collection of suitable sets by $\sscript (\Omega )$. The proof of the next theorem is similar to a proof given in \cite{gud112}.

\begin{thm}       % Theorem 3.1
\label{thm31}
$\sscript (\Omega )$ is a quadratic algebra and $\mutilde$ is a $q$-measure on $\sscript (\Omega )$ that extends
$\mu$ from $\cscript (\Omega )$.
\end{thm}

We call a real-valued function $f\in H$ a \textit{random variable} (Actually, we are considering random variables with finite second moment but this restriction is convenient here.) We now give a method for ``quantizing'' $f$ to obtain a bounded self-adjoint operator (\textit{observable}) $\fhat$ on $H$. Although we employ $\fhat$ to define a quantum integral of $f$, there may be another important use for $\fhat$. The map $f\mapsto\fhat$ transforms classical observables to quantum observables. If a discrete quantum process $\rho _n$ governs the dynamics for a discrete quantum gravity, then in some sense, Einstein's field equation should be an approximation to the sequence $\rho _n$ which gives a strong restriction on $\rho _n$. In this respect, the map $f\mapsto\fhat$ may be useful in transforming the observables of classical relativity to quantum relativity.

The \textit{quantization} of a nonnegative random variable $f$ is the operator $\fhat$ on $H$ defined by
\begin{equation*}
(\fhat g)(y)=\int\min\sqbrac{f(x),f(y)}g(x)d\nu _t(x)
\end{equation*}
It easily follows that $\doubleab{\fhat}\le\doubleab{f}$ so $\fhat$ is a bounded self-adjoint operator on $H$. If $f$ is an arbitrary random variable, we have that $f=f^+-f^-$ where $f^+(x)=\max\sqbrac{f(x),0}$ and $f^-=-\min\sqbrac{f(x),0}$. We define the bounded self-adjoint operator $\fhat$ on $H$ by $\fhat =f^{+\wedge}-f^{-\wedge}$. It can be shown that $\doubleab{\fhat}\le\max\paren{\doubleab{\fhat ^+},\doubleab{\fhat ^-}}$ \cite{gud112}. The next result summarizes some of the important properties of $\fhat$ \cite{gud112}.

\begin{thm}       % Theorem 3.2
\label{thm32}
{\rm (a)}\enspace For every $A\in\ascript$, $\chihat _A=\ket{\chi _A}\bra{\chi _A}$.
{\rm (b)}\enspace For every $\alpha\in\real$, $(\alpha f)^\wedge =\alpha\fhat$.
{\rm (c)}\enspace If $f\ge 0$, then $\fhat$ is a positive operator.
{\rm (d)}\enspace If $0\le f_1\le f_2\le\cdots$ is an increasing sequence of random variables converging in norm to a random variable $f$, then $\fhat _i\to\fhat$ in the operator norm topology.
{\rm (e)}\enspace If $f,g,h$ are random variables with disjoint supports, then
\begin{equation*}
(f+g+h)^\wedge = (f+g)^\wedge +(f+h)^\wedge +(g+h)^\wedge -\fhat -\ghat -\hhat
\end{equation*}
\end{thm}

Let $\rho$ be a state on $H$ and let $\mu _\rho$ be the corresponding $q$-measure on $\ascript _n$ or $\ascript$. If $f$ is a random variable, we define the $q$-\textit{integral} (or $q$-\textit{expectation}) of $f$ with respect to
$\mu _\rho$ as
\begin{equation*}
\int fd\mu _\rho =\rmtr (\rho\fhat )
\end{equation*}
The next corollary follows from Theorem~\ref{thm32}

\begin{cor}       % Corollary 3.3
\label{cor33}
{\rm (a)}\enspace For every $A\in\ascript$, $\int\chi _Ad\mu _\rho =\mu _\rho (A)$.
{\rm (b)}\enspace For every $\alpha\in\real$, $\int\alpha fd\mu _\rho =\alpha\int fd\mu _\rho$.
{\rm (c)}\enspace If $f\ge 0$, then $\int fd\mu _\rho\ge 0$.
{\rm (d)}\enspace If $f_i\ge 0$ is an increasing sequence of random variables converging in norm to a random variable $f$, then $\lim\int f_id\mu _\rho =\int fd\mu _\rho$.
{\rm (e)}\enspace If $f,g,h$ are random variables with disjoint supports, then
\begin{align*}
\int (f+g+h)d\mu _\rho&=\int (f+g)d\mu _\rho +\int (f+h)d\mu _\rho +\int (g+h)d\mu _\rho\\
  &\quad -\int fd\mu _\rho -\int gd\mu _\rho -\int hd\mu _\rho
\end{align*}
\end{cor}

The next result is called the \textit{tail-sum} formula and gives a justification for calling $\int fd\mu _\rho$ a
$q$-\textit{integral} \cite{gud09, gud111}.
\begin{equation*}
\end{equation*}

\begin{thm}       % Theorem 3.4
\label{thm34}
If $f\ge 0$ is a random variable, then
\begin{equation*}
\int fd\mu _\rho =\int _0^\infty\mu _\rho\paren{\brac{x\colon f(x)>\lambda}}d\lambda
\end{equation*}
where $d\lambda$ denotes Lebesgue measure on $\real$.
\end{thm}

It follows from Theorem~\ref{thm34} that if $f$ is an arbitrary random variable, then
\begin{equation*}
\int fd\mu _\rho =\int _0^\infty\mu _\rho\paren{\brac{x\colon f(x)>\lambda}}d\lambda
  -\int _0^\infty\mu _\rho\paren{\brac{x\colon f(x)<-\lambda}}d\lambda
\end{equation*}

Let $\rho _n$ be a discrete quantum process on $H$. We say that a random variable $f$ is \textit{integrable} for
$\rho _n$ if $\lim\rmtr (\rho _n\fhat )$ exists and is finite and in this case we define $\int fd\mutilde$ to be this limit. Notice that if $A\in\sscript (\Omega )$, then $\chi _A$ is integrable and $\int\chi _Ad\mutilde =\mutilde (A)$. The next result follows from Corollary~\ref{cor33}.

\begin{thm}       % Theorem 3.5
\label{thm35}
{\rm (a)}\enspace If $f$ is integrable and $\alpha\in\real$, then $\alpha f$ is integrable and
$\int\alpha fd\mutilde =\alpha\int fd\mutilde$.
{\rm (b)}\enspace If $f$ is integrable with $f\ge 0$, then $\int fd\mutilde\ge 0$.
{\rm (c)}\enspace If $f,g,h$ are integrable with mutually disjoint supports and $f+g$, $f+h$, $g+h$ are integrable, then $f+g+h$ is integrable and
\begin{align*}
\int (f+g+h)d\mutilde&=\int (f+g)d\mutilde +\int (f+h)d\mutilde +\int (g+h)d\mutilde\\
  &\quad -\int fd\mutilde -\int gd\mutilde -\int hd\mutilde
\end{align*}
\end{thm}

\section{Discrete Isometric Processes} % Section 4
Section~3 discussed a quantum gravity model in terms of a discrete quantum process $\rho _n$, $n=1,2,\ldots$, on
$H=L_2(\Omega ,\ascript ,\nu _t)$. It may be that $\rho _n$ is determined by a system of isometries (there is some controversy about whether this is possible \cite{sor03}). This would provide a restriction on the possible $\rho _n$. Moreover, we are familiar with dynamics governed by isometries so this might aid our intuition. The reader should note that such a formalism is motivated by and related to the sum over histories approach to quantum mechanics.
The results in this section are similar to results in \cite{gud112} taken in a different context.

Let $K_n$ be the Hilbert space of complex-valued function on $\pscript _n$ with the usual inner product
\begin{equation*}
\elbows{f,g}=\sum _{x\in\pscript _n}\overline{f(x)}g(x)
\end{equation*}
We call $K_n$ the $n$-\textit{site Hilbert space} and we denote the standard basis $\chi _{\brac{x}}$,
$x\in\pscript _n$, of $K_n$ by $e_x^n$. The projection operator $P_n(x)=\ket{e_x^n}\bra{e_x^n}$, $x\in\pscript _n$, describe the site at step $n$. In our context, a \textit{discrete isometric system} is a collection of isometries $U(s,r)$,
$r\le s\in\positive$, such that $U(s,r)\colon K_r\to K_s$, $U(r,r)=I_r$ and $U(t,r)=U(t,s)U(s,r)$ for every
$r\le s\le t\in\positive$. Recall that $U(s,r)$ is an isometry means that $U(s,r)$ is an operator satisfying
$U(s,r)^*U(s,r)=I_r$ and $U(s,r)U(s,r)^*=P_s$ where $I_r$ is the identity on $K_r$ and $P_s$ is the projection onto the range of $U(s,r)$ in $K_s$.

Let $U(s,r)$, $r\le s\in\positive$ be a discrete isometric system and let $\omega\in\Omega _n$ be an $n$-path. Since all $n$-paths go through $x_1$ of Figure~1 we can and shall assume that all $n$-paths begin at $x_1$. Then
$\omega$ has the form $\omega =x_1\omega _2\omega _3\ldots\omega _n$, $\omega _i\in\pscript _i$,
$i=1,2,\ldots ,n$. We describe $\omega$ by the operator $C_n(\omega )\colon K_1\to K_n$ given by
\begin{equation}         % equation (4.1)
\label{eq41}
C_n(\omega )=P_n(\omega _n)U(n,n-1)P_{n-1}(\omega _{n-1})U(n-1,n-2)\cdots P_2(\omega _2)U(2,1)
\end{equation}
Defining $a(\omega )$ by
\begin{align}         % equation (4.2)
\label{eq42}
a(\omega )&=\elbows{e_{\omega _n}^n,U(n,n-1)e_{\omega _{n-1}}^{n-1}}
  \elbows{e_{\omega _{n-1}}^{n-1},U(n-1,n-2)e_{\omega _{n-2}}^{n-2}}\notag\\
  &\quad\cdots\elbows{e_{\omega _2}^2,U(2,1)e_{x_1}^1}
\end{align}
\eqref{eq41} becomes
\begin{equation}         % equation (4.3)
\label{eq43}
C_n(\omega )=a(\omega )\ket{e_{\omega _n}^n}\bra{e_{x_1}^1}
\end{equation}
Of course, we can identify $K_1$ with $\complex$ so $C_n(\omega )$ is the operator given by
$C_n(\omega )\alpha =\alpha a(\omega )\ket{e_{\omega _n}^n}$ for every $\alpha\in\complex$. We call
$a(\omega )$ the \textit{amplitude} of $\omega\in\Omega _n$ and interpret $\ab{a(\omega )}^2$ as the probability of the path $\omega$ according to the dynamics $U(s,r)$. The next result shows that
$\omega\mapsto\ab{a(\omega )}^2$ is indeed a probability distribution.

\begin{lem}       % Lemma 4.1
\label{lem41}
For the $n$-path space $\Omega _n$ we have
\begin{equation*}
\sum _{\omega\in\Omega _n}\ab{a(\omega )}^2=1
\end{equation*}
\end{lem}
\begin{proof}
By \eqref{eq42} we have
\begin{align*}
\sum _{\omega\in\Omega _n}&\ab{a(\omega )}^2\\
   &=\sum _{\omega\in\Omega _n}\ab{\elbows{e_{\omega _n}^n,U(n,n-1)e_{\omega _{n-1}}^{n-1}}}^2
   \ab{\elbows{e_{\omega _{n-1}}^{n-1},U(n-1,n-2)e_{\omega _{n-2}}^{n-2}}}^2\\
   &\quad\cdots\ab{\elbows{e_{\omega _2}^2,U(2,1)e_{x_1}^1}}^2\\
   &=\sum _{\omega\in\Omega _{n-1}}\ab{\elbows{e_{\omega _{n-1}}^{n-1},U(n-1,n-2)e_{\omega _{n-2}}^{n-2}}}^2
   \cdots\ab{\elbows{e_{\omega _2}^2,U(2,1)e_{x_1}^1}}^2\\
&\quad\vdots\\
&=\sum _{\omega\in\Omega _2}\ab{\elbows{e_{\omega _2}^2,U(2,1)e_{x_1}^1}}^2=1\qedhere
\end{align*}
\end{proof}

The quantity $C_n(\omega ')^*C_n(\omega )$ describes the interference between the two paths
$\omega ,\omega '\in\Omega _n$. Applying \eqref{eq43} we see that
\begin{equation}         % equation (4.4)
\label{eq44}
C_n(\omega ')^*C_n(\omega )=\overline{a(\omega ')}a(\omega )\delta _{\omega _n,\omega '_n}I_1
\end{equation}
which we can identify with the complex number $\overline{a(\omega ')}a(\omega )\delta _{\omega _n,\omega 'n}$.
For $A\in\ascript _n$ the \textit{class operator} $C_n(A)$ is
\begin{equation*}
C_n(A)=\sum _{\omega\in A}C_n(\omega )
\end{equation*}
It is clear that $A\mapsto C_n(A)$ is an operator-valued measure on the algebra $\ascript _n$. Moreover,
$C_n(\Omega _n)=U(n,1)$ because by \eqref{eq42} and \eqref{eq43} we have
\begin{align*}
C_n(\Omega _n)&=\sum _{\omega\in\Omega _n}C_n(\omega )=\sum _{\omega\in\Omega _n}
\elbows{e_{\omega _n}^n,U(n,1)e_{x_1}^1}\ket{e_{\omega _n}^n}\bra{e_{x_1}^1}\\
&=U(n,1)
\end{align*}
It is well-known that $D_n\colon\ascript _n\times\ascript _n\to\complex$ defined by
\begin{equation*}
D_n(A,B)=\elbows{C_n(A)^*C_n(B)e_{x_1}^1,e_{x_1}^1}
\end{equation*}
is a decoherence functional and we see that $D_n(\Omega _n,\Omega _n)=1$. Defining the $q$-measure
$\mu _n\colon\ascript _n\to\real ^+$ by $\mu _n(A)=D_n(A,A)$, we have that $\mu _n(\Omega _n)=1$.

The $n$-\textit{distribution} on $\pscript _n$ given by
\begin{equation*}
p_n(x)=\mu _n\paren{\brac{\omega\in\Omega _n\colon\omega _n=x}}
\end{equation*}
is interpreted as the probability that site $x$ is visited at step $n$. The next result shows that $p_n$ gives the usual quantum distribution.

\begin{thm}       % Theorem 4.2
\label{thm42}
For $x\in\pscript _n$ we have
\begin{equation*}
p_n(x)=\ab{\sum _{\omega _n=y}a(\omega )}^2=\ab{\elbows{e_n^n,U(n,1)e_{x_1}^1}}^2
\end{equation*}
\end{thm}
\begin{proof}
Letting $A=\brac{\omega\in\Omega _n\colon\omega _n=x}$ we have by \eqref{eq44} that
\begin{align*}
p_n(x)=D_n(A,A)&=\elbows{C_n(A)^*C_n(A)e_{x_1}^1,e_{x_1}^1}\\
   &=\sum\brac{\elbows{C_n(\omega ')^*C_n(\omega )e_{x_1}^1,e_{x_1}^1}\colon\omega '_n=\omega _n=x}\\
   &=\sum\brac{\overline{a(\omega ')}a(\omega )\colon\omega '_n=\omega _n=x}\\
   &=\ab{\sum _{\omega _n=x}a(\omega )}^2
\end{align*}
By \eqref{eq42} we have that
\begin{equation*}
\sum _{\omega _n=x}a(\omega )=\elbows{e_x^n,U(n,1)e_{x_1}^1}
\end{equation*}
and the result follows.
\end{proof}

We define the \textit{decoherence matrix} as the matrix $\dhat _n$ with components
\begin{equation*}
\dhat _n(\omega ,\omega ')=D_n\paren{\brac{\omega},\brac{\omega '}}
\end{equation*}
$\omega ,\omega '\in\Omega _n$. We have by \eqref{eq44} that
\begin{equation}         % equation (4.5)
\label{eq45}
\dhat _n(\omega ,\omega ')=\elbows{C_n(\omega ')^*C_n(\omega )e_{x_1}^1,e_{x_1}^1}
  =a(\omega )\overline{a(\omega ')}\delta _{\omega _n,\omega '_n}
\end{equation}
Notice that $\mu _n(\omega )=\dhat _n(\omega ,\omega )=\ab{a(\omega )}^2$ and by Lemma~\ref{lem41} that 
$\sum _{\omega\in\Omega _n}\mu _n(\omega )=1$. Finally, notice that
\begin{align*}
D_n(A,B&=\sum\brac{\dhat _n(\omega ,\omega ')\colon\omega\in A,\omega '\in B}\\
&=\sum\brac{a(\omega )\overline{a(\omega ')}\delta _{\omega _n,\omega '_n}\colon\omega\in A,\omega '\in B}
\end{align*}
and hence
\begin{equation}         % equation (4.6)
\label{eq46}
\mu _n(A)=D_n(A,A)=\sum _{\omega ,\omega '\in A}\dhat _n(\omega ,\omega ')
  =\sum _{\omega ,\omega '\in A}a(\omega )\overline{a(\omega ')}\delta _{\omega _n,\omega '_n}
\end{equation}

Define the Hilbert space $H'_n$ as the set of complex-valued functions on $\Omega _n$ with the usual inner product
\begin{equation*}
\elbows{f,g}=\sum _{\omega\in\Omega _n}\overline{f(\omega )}g(\omega )
\end{equation*}
Then $\dhat _n$ corresponds to the operator (also denoted by $\dhat _n$) given by
\begin{equation*}
(\dhat _nf)(\omega )=\sum _{\omega '\in\Omega _n}\dhat _n(\omega ,\omega ')f(\omega ')
\end{equation*}

\begin{thm}       % Theorem 4.3
\label{thm43}
The operator $\dhat _n$ is a state on $H'_n$.
\end{thm}
\begin{proof}
It follows from \eqref{eq45} that $\dhat _n$ is a positive operator \cite{gsapp, mocs05}. By Lemma~\ref{lem41} we have
\begin{equation*}
\rmtr (\dhat)=\sum _{\omega \in\Omega _n}D_n(\omega ,\omega )
  =\sum _{\omega\in\Omega _n}\ab{a(\omega )}^2=1
\end{equation*}
Hence, $\dhat _n$ is a trace~1 positive operator so $\dhat _n$ is a state on $H'_n$.
\end{proof}

Each $\omega\in\Omega _n$ corresponds to a unit vector $\chi _{\brac{x}}$ in $H'_n$ and for every $A\in\ascript _n$ we have the vector $\ket{\chi _A}=\sum\brac{\chi _{\brac{x}}\colon\omega\in A}$.

\begin{lem}       % Lemma 4.4
\label{lem44}
The decoherence functional satisfies
\begin{equation*}
D_n(A,B)=\rmtr\paren{\ket{\chi _B}\bra{\chi _A}\dhat _n}
\end{equation*}
\end{lem}
\begin{proof}
For every $A,\in\ascript _n$ we have
\begin{align*}
\rmtr\paren{\ket{\chi _B}\bra{\chi _A}\dhat _n}&=\sum _{\omega\in\Omega _n}
  \elbows{\ket{\chi _B}\bra{\chi _A}\dhat _n\chi _{\brac{\omega}},\chi _{\brac{\omega}}}\\
  &=\sum _{\omega\in\Omega _n}
  \elbows{\dhat _n\chi _{\brac{\omega}},\ket{\chi _A}\bra{\chi _B}\chi _{\brac{\omega}}}\\
  &=\sum\brac{\elbows{\dhat _n\chi _{\brac{\omega}},\chi _A}\colon\omega\in B}\\
  &=\sum\brac{\elbows{\dhat _n\chi _{\brac{\omega}},\chi _{\brac{\omega '}}}\colon\omega\in B,\omega '\in A}\\
  &=\sum\brac{D_n(\omega ',\omega )\colon\omega '\in A,\omega\in B}=D_n(A,B)\qedhere
\end{align*}
\end{proof}

We now transfer the states $\dhat _n$ on $H'_n$ to states on $H_n$, $n=1,2,\ldots\,$. The set
$\brac{\chi _{\brac{\omega}}\colon\omega\in\Omega _n}$ forms an orthonormal basis for $H'_n$ and assuming that $p_t^n\paren{\brac{\omega}}\ne 0$ for all $\omega\in\Omega _n$, we have that
$\brac{v_\omega\colon\omega\in\Omega _n}$ is an orthonormal basis for $H_n$ where
\begin{equation*}
v_\omega =p_t^n\paren{\brac{\omega}}^{-1/2}\chi _{\rmcyl (\omega )}
\end{equation*}
Defining $U_n\chi _{\brac{\omega}}=v_\omega$ and extending $U_n$ by linearity, $U_n\colon H'_n\to H_n$ becomes a unitary operator and $U_n\colon H'_n\to H$ is an isometry from $H'_n$ into $H$. Letting $P_n$ be the projection of $H$ onto the subspace $H_n$ we have
\begin{equation*}
P_nf=\sum _{\omega\in\Omega _n}\elbows{v_\omega ,f}v_\omega
  =\sum _{\omega\in\Omega _n}p_t^n
  \paren{\brac{\omega}}^{-1}\int f\chi _{\rmcyl (\omega )}d\nu _t\chi _{\rmcyl (\omega )}
\end{equation*}
In particular, for $A\in\ascript$ we obtain
\begin{equation*}
P_n\chi _A=\sum _{\omega\in\Omega _n}p_t^n
\paren{\brac{\omega}}^{-1}\nu _t\paren{A\cap\rmcyl (\omega )}\chi _{\rmcyl (\omega )}
\end{equation*}
Hence,
\begin{equation*}
P_n1=\sum _{\omega\in\Omega _n}\chi _{\rmcyl (\omega )}=1
\end{equation*}
To transfer $\dhat _n$ from $H'_n$ to $H_n$ we define $\rho _n=U_n\dhat _nU_n^*P_n$. Then $\rho _n$ is a state on $H_n$ and also on $H$ as before.

\begin{thm}       % Theorem 4.5
\label{thm45}
The sequence of states $\rho _n$, $n=1,2,\ldots$, is consistent.
\end{thm}
\begin{proof}
To show that $\rho _n$ is consistent is equivalent to showing that
\begin{equation}         % equation (4.7)
\label{eq47}
D_n(A\times\pscript _{n+1},B\times\pscript _{n+1})=D_n(A,B)
\end{equation}
for all $A,B\in\ascript _n$. Using the notation $\omega x=\omega _1\omega _2\cdots\omega _n x$ for
$\omega\in\Omega _n$ and $x\in\pscript _{n+1}$, \eqref{eq47} is equivalent to
\begin{equation*}
D_n(\omega ,\omega ')=\sum _{x,y\in\pscript _{n+1}}D_{n+1}(\omega x,\omega 'x)
=\sum _{x\in\pscript _{n+1}}D_{n+1}(\omega x,\omega 'x)
\end{equation*}
for every $\omega ,\omega '\in\Omega _n$. Since
\begin{equation*}
\sum _{x\in\pscript _{n+1}}\elbows{U(n+1,n)e_{\omega '_n}^n,e_x^{n+1}}
\elbows{e_n^{n+1},U(n+1,n)e_{\omega _n}^n}=\delta _{\omega _n,\omega '_n}
\end{equation*}
 it follows that
 \begin{equation*}
 \sum _{x\in\pscript _{n+1}}a(\omega x)\overline{a(\omega 'x)}
 =a(\omega )\overline{a(\omega ')}\delta _{\omega _n,\omega '_n}
\end{equation*}
for all $\omega ,\omega '\in\Omega _n$. Hence
\begin{align*}
\sum _{x\in\pscript _{n+1}}D_{n+1}(\omega x,\omega 'x)
  &=\sum _{x\in\pscript _{n+1}}a(\omega x)\overline{a(\omega 'x)}
  =a(\omega )\overline{a(\omega ')}\delta _{\omega _n,\omega '_n}\\
  &=D_n(\omega ,\omega ')\qedhere
\end{align*}
\end{proof}

We conclude from Theorem~\ref{thm45} that $\rho _n$ is a discrete quantum process on
$H=L_2(\Omega ,\ascript ,\nu _t)$ that was constructed from a discrete isometric process.

\end{document}